\documentclass[a4paper,11pt]{article}

\usepackage[utf8]{inputenc}
\usepackage[english]{babel}
\usepackage{amssymb,amsmath,amsthm}

\newtheorem{theorem}{Theorem}
\newtheorem{corollary}{Corollary}
\usepackage[margin=0.95in]{geometry}

\usepackage{amsthm,enumerate}
\usepackage{graphicx}
\usepackage{color,hyperref}
\usepackage{tabularx}
\usepackage{amsmath}
\usepackage{amssymb,amsfonts,textcomp}
\usepackage{latexsym}
\usepackage{graphicx} 
\usepackage{multirow}
\usepackage{microtype}
\usepackage{rotating}
\usepackage{url}
\usepackage{color}
\usepackage{comment}
\usepackage{booktabs}
\usepackage{bm}
\usepackage[]{units}
\usepackage[noend, ruled, noline]{algorithm2e}
\usepackage{tikz}
\usepackage{tcolorbox,complexity,multirow}
\usepackage{float}

\usepackage{algorithm2e}
\SetKw{Continue}{continue}

\newtheorem{definition}{Definition}[section]
\newtheorem{prop}[definition]{Proposition}

\newtheorem{lemma}[definition]{Lemma}
\newtheorem{claim}[definition]{Claim}

\newtheorem{observation}[definition]{Observation}

\newcommand{\vote}{\mathsf{vote}}

\newcommand{\runtime}{$\mathcal{O}(|E|)$}
\newcommand{\runtimeold}{$\mathcal{O}(\sqrt{|V|\alpha(|V|,|E|)} |E| \log ^\frac{3}{2}|V|)$}

\usetikzlibrary{arrows,decorations.markings,patterns,calc,positioning,backgrounds}
\tikzstyle{vertex} = [circle, draw=black, fill=black, scale= 0.5]
\tikzstyle{edgelabel} = [circle, fill=white, scale= 0.9]
\tikzstyle{arrow} = [line width=0.8mm,-implies,double, double distance=0.8mm]
\tikzstyle{dashedpointedline} = [line width=0.2mm,dashed,dash pattern=on 2mm off 1mm,
decoration={markings,
	mark=at position 1 with {\arrow[line width=0.2mm, scale= 1.8]{>}}
},
postaction={decorate}
] 
\tikzstyle{pointedline} = [line width=0.3mm,
decoration={markings,
	mark=at position 1 with {\arrow[line width=0.2mm, scale= 2]{>}}
},
postaction={decorate}
] 

\makeatletter
\newcommand\newtag[2]{#1\def\@currentlabel{#1}\label{#2}}
\makeatother

\title{Testing popularity in linear time via maximum matching}

\author{Erika Bérczi-Kovács\\ ELTE Eötvös Loránd University, Budapest, Hungary
\\HUN-REN–ELTE Egerváry Research Group on Combinatorial Optimization 
\and
Kata Kosztolányi\\
 ELTE Eötvös Loránd University, Budapest, Hungary}

\begin{document}

\maketitle

\begin{abstract}
Popularity is an approach in mechanism design to find fair structures in a graph, based on the votes of the nodes.
Popular matchings are the relaxation of stable matchings: given a graph $G=(V,E)$ with strict preferences on the neighbors of the nodes, a matching $M$ is popular if there is no other matching $M'$ such that the number of nodes preferring $M'$ is more than those preferring $M$. This paper considers the popularity testing problem, when the task is to decide whether a given matching is popular or not.
Previous algorithms applied reductions to maximum weight matchings.
We give a new algorithm for testing popularity by reducing the problem to maximum matching testing, thus attaining a linear running time \runtime.

Linear programming-based characterization of popularity is often applied for proving the popularity of a certain matching. As a consequence of our algorithm we derive a more structured dual witness than previous ones. Based on this result we give a combinatorial characterization of fractional popular matchings, which is a special class of popular matchings.
\end{abstract}

\section{Introduction}
The notion of popularity was introduced by G\"ardenfors \cite{Gar75}. In a popular roommates problem, a graph $G = (V, E)$ is given with strict preferences ($\succ_v$) over the neighbors of~$v$ for each node $v\in V$, similarly to the stable roommates problem \cite{Irv85}. 
A node prefers a matching $M$ over matching $M'$ if it prefers its partner in $M$ compared to the one in $M'$  (a node prefers every neighbor over being unmatched).
We say that a matching $M$ is \textbf{more popular} than $M'$ if the number of nodes preferring $M$ to $M'$ is larger than the number of those preferring $M'$ to~$M$. A matching $M$ is called \textbf{popular} if there is no matching~$M'$ that is more popular than~$M$.

Since popularity is closely related to stability, we also define stable matchings: given a matching $M$, an edge  $uv \in E\setminus M$ is called a \textbf{blocking edge} if both $v \succ_u M(u)$ and $u \succ_v M(v)$, that is, both $u$ and $v$ prefer each other over their partners in $M$. A matching $M$ is \textbf{stable} if there is no blocking edge in $G$. 

As for the relation of stability and popularity, G\"ardenfors \cite{Gar75} showed that every stable matching is also popular in a bipartite graph. Later Chung \cite{Chu00} proved the same for the general case. Moreover, a stable matching is also a minimum size popular matching, shown by Huang and Kavitha \cite{BIM10}. So popularity can be regarded as a relaxation of stability of a matching, where we relax the constraint of local stability to a global one. 
The motivation behind considering a relaxation of stability is twofold: first, there are instances of the stable roommates problem where no stable matching exists, but a popular one can be given \cite{BIM10}. Second, the size of a popular matching can be greater than that of a stable one, thus including more participants in a matching which is fair from a global pont of view.
The maximum-size popular matching problem can be solved in polynomial time in bipartite graphs \cite{HK13,Kav14}. 
Although deciding the existence of a stable matching can be done in polynomial time (Irving, \cite{Irv85}), deciding the existence of a popular matching is NP-complete, shown by Faenza et al. \cite{FKPZ19} and Gupta et al. \cite{GMSZ21}.
For further reading on popularity see \cite{Cse17}, \cite{Man13}.

 In this paper we investigate the popularity testing problem, when the task is to decide whether a given matching is popular or not. 
To the best of our knowledge, \runtimeold$ $ (where $\alpha$ is the inverse Ackermann’s function) is the running time of the current fastest algorithm for testing popularity, given by Biró et al. \cite{BIM10}, whereas in bipartite graphs (i.e. popular marriage problem)  testing popularity can be done in \runtime \; (Huang and Kavitha \cite{HK13}).
It is a natural question to ask whether a linear time algorithm can be given in the non-bipartite case as well. 
We show that linear  running time is achievable in that case (Theorem \ref{thm:G*}).

We also consider the linear programming-based characterization of popularity, which can be applied in several context to show the popularity of a matching, see for example \cite{KMN09}, \cite{DBLP:conf/icalp/Kavitha16} or \cite{HK13}. It is an important aspect in such applications to what extent the values of a dual witness can be restricted. We investigate this question also, and give a more structured dual witness than previous ones. Finally we show an application of our result for fractional popular matchings \cite{Kav19}, which are a special class of popular matchings. 

\subsection*{Our contribution}

\begin{enumerate}
    \item We give a new algorithm for testing popularity. Previous approaches apply maximum weight matching for the problem, whereas we reduce the problem to testing the maximality of a matching. Thus the running time of our algorithm is \runtime, which is the best one may expect (Theorem \ref{thm:G*}). 
    \item Also, as a consequence of our approach an LP dual witness can be derived with values $\{-1,0,+1\}$ on nodes and only $0$ or $2$ on odd-sets, which is more restrictive than previous results (Theorem \ref{thm:0-2}).
    \item Applying the previous dual witness, we give a combinatorial characterization for a matching to be popular in the fractional sense (Theorem \ref{thm:char-trula-popular}). This leads to a combinatorial algorithm for testing this property.
\end{enumerate}

The rest of the paper is organized as follows: in Section \ref{sec:previous} we summarize previous results and techniques on popularity testing. In Section \ref{sec:reduction} we prove our main result, reducing the popularity testing problem to a maximum matching testing problem. Finally in Section \ref{sec: polyhedra} we prove a theorem on a structured dual witness and show an application for fractional popular matchings \ref{subsec:frac}.

\section{Previous results}\label{sec:previous}

In this section we summarize previous results on characterizing and testing popularity.
There are two characterizations for the popularity of a matching, one applying a reduction to maximum weight perfect matching and one based on excluding certain alternating structures. Previous popularity testing algorithms rely on the former one.
In Subsection~\ref{subsec:mwchar} we present the maximum weight matching approach and related testing algorithms (\ref{subsubsec:previous algos}) and previous results on dual witness (\ref{subsubsec: previous witness}). In Subsection~\ref{subsec:alterneting char} the alternating structure based characterization is described, which will be the starting point of our algorithm.

\subsubsection*{Some notations}

For the rest of the paper, let a popular roommates instance be given on graph $G$ and let $M$ be a matching in~$G$.

For a subset of edges $F$, let $V(F)$ denote the set of those nodes that have at least one incident edge from $F$. For example, $V(M)$ is the set of nodes matched by $M$.
Let $M(v)$ denote the pair of $v$ in $M$.
For a subset of nodes $X$, let $G[X]$ denote the subgraph of $G$ spanned by $X$.

Let $x$ be a vector $x\in \mathbb{R}^E$. For a node $u\in V$ let $d^G_x(u):=\sum_{uv \in E} x_{uv}$, and for a subset of nodes $Z\subseteq V$ let $i^G_x(Z)$ denote the sum of $x_e$ values over edges spanned by $Z$, that is, $i^G_x(Z):=\sum_{uv \in G[Z]} x_{uv}$. 

For a graph $G=(V,E)$, let $\mathbb{B}(G)$ denote the set of odd subsets of the nodes.

\subsection{Maximum weight perfect matching-based characterization of popularity}\label{subsec:mwchar}

Kavitha et al. observed that for popular marriages the testing problem can be reduced to maximum weight perfect matchings \cite{KMN09}. Biró et al.\cite{BIM10} gave an analogous reduction for the roommates problem by proving that testing the popularity of a matching $M$ can be reduced to testing whether the maximum weight of a perfect matching is $0$ in a proper weighted auxiliary graph. Here we describe a similar reduction given by Kavitha \cite{Kav19}.

Let $\widetilde{G}$ denote the graph derived from $G$ by adding loops to every node, and for a matching $M$ let $\widetilde{M}$ denote the perfect matching in $\widetilde{G}$ derived from $M$ by adding loops for all unmatched nodes. 
For a node $u\in V$ and its neighbors $v,w$ in $\widetilde{G}$ let $\vote_u(v,w)$ denote the preference of $u$ over $v$ and $w$: $$\vote_u(v,w):=\left\{\begin{array}{rl}
		+1 & \textrm{if $v\succ_u w$}\\
		-1 & \textrm{if $w\succ_u v$}\\
		0  & \textrm{if $v=w$} \end{array}\right.$$
For example, $\vote_u(v,u)=1$ for every neighbor $v$ of $u$. 

Now we define edge weights $w_M$ on $\widetilde{E}$ for matching $M$. The weight of an edge $uv\in E$ is $w_M(uv):=\vote_u(v,\widetilde{M}(u))+\vote_v(u,\widetilde{M}(v))$, where loops have themselves as pairs. 
So the weight of an edge can be $-2, 0$ or $2$, and the ones with value $2$ are exactly the blocking edges.
For a loop $uu$ we define $w_M(uu):=\vote_u(u,\widetilde{M})$, that is, $$w_M(uu)=\left\{\begin{array}{rl}
		0 & \textrm{if $u \notin V(M)$}\\
		-1 & \textrm{if $u \in V(M)$}\end{array}\right.$$\label{wdef}
Note that $\widetilde{M}$ is always a $0$-valued perfect matching in $\widetilde{G}$.

We can reduce popularity to maximum weight perfect matching the following way.
\begin{theorem}[Kavitha, \cite{Kav19}]\label{thm:char-popular-mwpm}
Matching $M$ is popular in $G$ if and only if the maximum $w_M$-weight of a perfect matching in $\widetilde{G}$ is $0$.
\end{theorem}

\subsubsection{Previous results for testing popularity}\label{subsubsec:previous algos}

 Biró et al. \cite{BIM10} used a similar construction as in \cite{Kav19} and gave an \runtimeold \; algorithm for testing popularity (where $\alpha$ is the inverse Ackermann’s function). Their solution is based on the observation that their reduction is a special case of a maximum weight perfect matching problem with weights  $\{-1,0,1\}$, so the algorithm of Gabow and Tarjan \cite{GT91} can be applied.  
In their paper Biró et al. ask whether it is possible to check popularity with a better running time.

\subsubsection{Previous results on dual witness}\label{subsubsec: previous witness}

According to Theorem \ref{thm:char-popular-mwpm} the popularity of a matching can be characterized as a (zero-weight) maximum weight perfect matching in $\widetilde{G}$. Its polyhedral description is the following (see e.g. LP1 in \cite{Kav19}):

\begin{alignat}{5}
\text{LP1:}&&&&&\quad\label{LP1}\notag\\
&&\text{maximize}    &\sum_{e\in E(\widetilde{G})}& w_M(e) \cdot x_e   &    \quad\nonumber\\
&&\text{subject to}  && \quad &\notag\\
&&&d^{\widetilde{G}}_x(v)=1&  \quad&\forall v\in V\nonumber\\
&&& i^G_x(Z)\leq &(|Z|-1)/2         \quad  \quad&\forall Z\in \mathbb{B}(G)\nonumber\\
&&& x_e \geq 0&  \quad&\forall e \in E(\widetilde{G})\nonumber
\end{alignat}
The dual of LP1 is the following.
\begin{alignat}{4}
\text{LP2}:&&&&\notag\\
&\text{minimize}    &\sum_{v\in V}&\alpha_v+\sum_{Z\in\mathbb{B}(G)}(|Z|-1)/2  \cdot y_Z   \quad   &\nonumber\\
&\text{subject to}  && \quad &\notag\\
&&\alpha_v+\alpha_w+\sum_{v,w\in Z, Z\in \mathbb{B}(G)}y_Z&\geq w_M(vw) \quad & \forall vw\in E(G)\nonumber\\
&& \alpha_v&\geq w_M(vv)  & \forall v\in V\nonumber\\
&& y_Z &\geq 0 \quad &\forall Z\in\mathbb{B}(G)\nonumber 
\end{alignat}

For a popular matching $M$, an optimal dual solution $(\alpha, y)$ is called a \textbf{dual witness}. 
In \cite{Kav19} Kavitha showed that for a popular matching $M$ a dual witness with restricted values can be given.

\begin{theorem}[Kavitha \cite{Kav19}, Lemma 7]
Let $M$ be a popular matching in $G$. Then $M$ has a dual witness $(\alpha, y)$ for LP2 such that $\alpha_v \in \{-1,0,1\}$ for all $ v\in V$ and $y_Z \in \{0,1,2\}$ for all $ Z\in\mathbb{B}(G)$.
\end{theorem}

We show in Section \ref{sec: polyhedra} that it is enough to allow values $\{0,2\}$ on odd sets (Theorem \ref{thm:0-2}).

\subsection{Alternating structure-based characterization of popularity}\label{subsec:alterneting char}

In this subsection we describe the characterization of popularity which will be used in our algorithm.

A \textbf{walk} in $G$ is a sequence of nodes $v_1,v_2,\ldots,v_k$ such that $v_i,v_{i+1}$ is an edge of $G$ for each $1\leq i\leq k-1$.
A \textbf{path} in $G$ is a walk with \emph{distinct} nodes. A \textbf{cycle} is also a walk, such that $v_1=v_k$ and all other pairs of nodes are distinct.  A  walk is \textbf{alternating} with respect to matching $M$, if its alternate edges belong to $M$. 

Consider edge weights $w_M$ defined in Section~\ref{subsec:mwchar}, and let $G_M$ denote the subgraph of $G$ derived by deleting all $uv$ edges with $w_M(uv)=-2$. 

\begin{theorem}[Huang and Kavitha~\cite{HK13}]
  \label{thm:char-popular}
$M$ is popular in $G$ if and only if $G_M$ does not contain any of the following with respect to~$M$:
	\begin{enumerate}[i)]
		\item an alternating cycle with a blocking edge;
		\item an alternating path connecting two disjoint blocking edges;
		\item an alternating path connecting a blocking edge with an unmatched node as an end node.
	\end{enumerate}
\end{theorem}

The structures above proving unpopularity of a matching will be called \textbf{blocking structures}. The cycle in case i) is a \textbf{blocking cycle} and paths in cases ii) and iii) are \textbf{blocking paths}.

\section{Reduction to testing maximum matching}\label{sec:reduction}

We will show that popularity testing can be reduced to testing whether a given matching in a graph is maximum or not.
In contrast to previous approaches, our algorithm relies on the characterization based on excluded blocking structures (see Theorem \ref{thm:char-popular}).

Let $B\subset E$ denote the set of blocking edges.
Consider the set of nodes $V(B)$ incident to blocking edges.
A node $v \in V(B)$ is called a \textbf{leaf node} if there is only one blocking edge $uv$ incident to $v$.
If there are at least two leaf nodes $v_1,v_2,\ldots$ with blocking edges $zv_i$ incident to the same neighbor $z$, they form a \textbf{star} $S=\{zv_1,zv_2,\ldots,zv_k\}$  ($k\geq 2$) with \textbf{middle node} $z$. Note that $z$ may have other incident blocking edges connecting $z$ to non-leaf nodes. 

We will define an auxiliary graph $G_M^*=(V^*,E^*)$. We take $G_M$ as a starting graph for $G_M^*$ and make the following steps (see Fig. \ref{fig:G_M^*}). 

\begin{enumerate}
    \item For every node $v \in V(B)$ which is not a leaf node in a star we add an extra node $b_v$ and edge $v b_v$. (Such a node $b_v$ is called a \textbf{blocking node}.)
    \item For leaf nodes in a star $S=\{zv_1,zv_2,\ldots,zv_k\}$ we add one common extra node $b_S$ and edges $v_ib_S$ ($1\leq i\leq k$). (Node $b_S$ is called a \textbf{star node}.)
    \item We create a new node $u$, and contract the set of unmatched nodes $V\setminus V(M)$ to~$u$.
    \item We delete the set $B$ of blocking edges from $G_M^*$.

\end{enumerate}
For easier navigation between $G$ and $G_M^*$ we define a function $b:V(B)\to V^*$:
\begin{equation}
b(v)=
\begin{cases}
      b_v& \text{ if } v \text{ is not a leaf in a star } \\
      b_S& \text{ if } v \text{ is a leaf in star } $S$.
\end{cases}
\end{equation} 
 Note that matching $M$ is also a matching in $G_M^*$ and the set of unmatched nodes in $G_M^*$ is $b(V(B))\cup\{u\}$.

 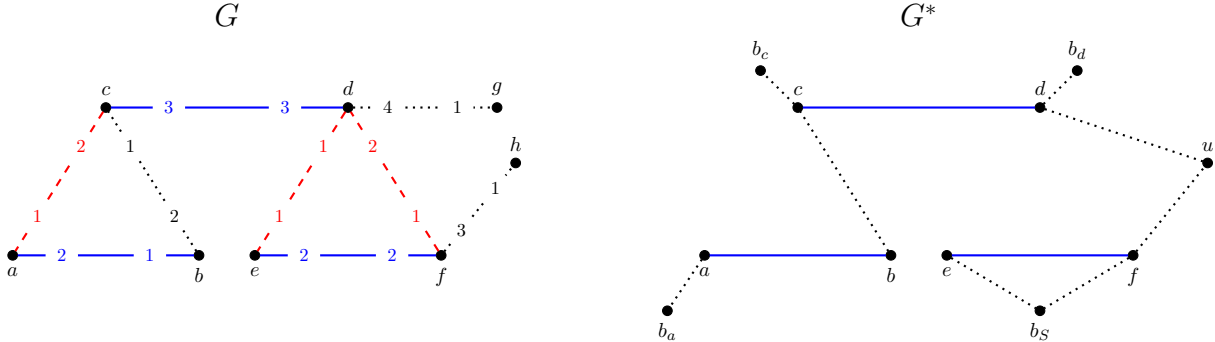
\begin{figure}[ht]
\begin{tikzpicture}[scale=0.7, transform shape]

  \pgfmathsetmacro{\d}{3.5}	
  \pgfmathsetmacro{\b}{3}
  \pgfmathsetmacro{\D}{13}
		
    \node[vertex, label=below:$a$] (A1) at (0,0) {};
	\node[vertex, label=below:$b$] (A2) at ($(A1) + (\d, 0)$) {};
	\node[vertex, label=above:$c$] (A3) at ($(A1) + (0.5*\d, 0.8*\d)$) {};

	\draw [thick, color=blue] (A1) -- node[edgelabel, near start] {2} node[edgelabel, near end] {1} (A2);
	\draw [thick, dotted] (A2) -- node[edgelabel, near start] {2} node[edgelabel, near end] {1} (A3);
	\draw [thick,  dashed, color=red] (A3) -- node[edgelabel, near start] {2} node[edgelabel, near end] {1} (A1);

    \node[] (G) at (1.15*\d, 1.3*\d) {\LARGE{$G$}};

     \node[vertex, label=below:$e$] (B1) at ($(A1) + (\d*1.3, 0)$) {};
	\node[vertex, label=below:$f$] (B2) at ($(B1) + (\d, 0)$) {};
	\node[vertex, label=above:$d$] (B3) at ($(B1) + (0.5*\d, 0.8*\d)$) {};
	\node[vertex, label=above:$g$] (B4) at ($(B1) + (1.3*\d, 0.8*\d)$) {};
    \node[vertex, label=above:$h$] (B5) at ($(B1) + (1.4*\d, 0.5*\d)$) {};

	\draw [thick, color=blue] (B1) -- node[edgelabel, near start] {2} node[edgelabel, near end] {2} (B2);
	\draw [thick, dashed, color=red] (B2) -- node[edgelabel, near start] {1} node[edgelabel, near end] {2} (B3);
	\draw [thick,  dashed, color=red] (B3) -- node[edgelabel, near start] {1} node[edgelabel, near end] {1} (B1);

    \draw [thick, color=blue] (B3) -- node[edgelabel, near start] {3} node[edgelabel, near end] {3} (A3);

    \draw [thick, dotted] (B3) -- node[edgelabel, near start] {4} node[edgelabel, near end] {1} (B4);

    \draw [thick, dotted] (B2) -- node[edgelabel, near start] {3} node[edgelabel, near end] {1} (B5);

    \node[vertex, label=below:$a$] (A1) at (\D,0) {};
	\node[vertex, label=below:$b$] (A2) at ($(A1) + (\d, 0)$) {};
	\node[vertex, label=above:$c$] (A3) at ($(A1) + (0.5*\d, 0.8*\d)$) {};
	\node[vertex, label=above:$b_c$] (A4) at ($(A1) + (0.3*\d, 1*\d)$) {};
	\node[vertex, label=below:$b_a$] (A5) at ($(A1) + (-0.2*\d, -0.3*\d)$) {};

	\draw [thick, color=blue] (A1) --  (A2);
	\draw [thick, dotted] (A2) -- (A3);

    \node[] (G*) at ($(A1) + (1.15*\d, 1.3*\d)$) {\LARGE{$G_M^*$}};

     \node[vertex, label=below:$e$] (B1) at ($(A1) + (1.3*\d, 0)$) {};
	\node[vertex, label=below:$f$] (B2) at ($(B1) + (\d, 0)$) {};
	\node[vertex, label=above:$d$] (B3) at ($(B1) + (0.5*\d, 0.8*\d)$) {};
    \node[vertex, label=above:$b_d$] (B4) at ($(B1) + (0.7*\d, 1*\d)$) {};
     \node[vertex, label=below:$b_S$] (B5) at ($(B1) + (0.5*\d, -0.3*\d)$) {};

    \node[vertex, label=above:$u$] (B6) at ($(B1) + (1.4*\d, 0.5*\d)$) {};

	\draw [thick, color=blue] (B1) --  (B2);

 \draw [thick, color=blue] (B3) -- (A3);
 \draw [thick, dotted] (A1) -- (A5);
 \draw [thick, dotted] (A3) -- (A4);
 \draw [thick, dotted] (B3) -- (B4);
 \draw [thick, dotted] (B1) -- (B5);
 \draw [thick, dotted] (B2) -- (B5);
 \draw [thick, dotted] (B3) -- (B6);
 \draw [thick, dotted] (B2) -- (B6);

\end{tikzpicture}
\caption{An example for $G_M^*$. Matching $M=\{ab, cd, ef\}$, blocking edges are $\{ac,de,df\}$. There is one star $S=\{de,df\}$ with middle node $d$. There are several options to see that $M$ is not popular. First, there is an augmenting path $b_c-c-d-b_d$ in $G_M^*$ connecting $b_c$ and $b_d$, thus there is an alternating path $a-c-d-f$ connecting blocking edges $ac$ and $df$, so matching $M'=\{ac,df\}$ is more popular than $M$. There is another alternating path $b_c-c-d-u$ in $G_M^*$ giving that matching $M''=\{ac,dg,ef\}$ is also more popular than $M$. Finally alternating path $b_S-e-f-u$ shows that matching $M'''=\{ab,de,fh\}$ is also more popular than $M$.}\label{fig:G_M^*}
\end{figure}

\medskip
The following theorem connects popularity of $M$ in $G$ with maximality of $M$ in $G_M^*$.
\begin{theorem}\label{thm:G*}
Let an instance of the popular roommates problem be given on graph $G=(V,E)$ with preferences $\succ_v$ and let $G_M^*$ be the auxiliary graph defined above for matching $M$. Then $M$ is popular in $G$ if and only if $M$ is maximum size in $G_M^*$.
Testing the popularity of $M$ can be done in \runtime.
\end{theorem}
\begin{proof}
We will use the well-known characterization that a matching is maximum size if and only if there is no alternating path connecting two unmatched nodes.

$\Leftarrow:$

According to the characterization in Theorem \ref{thm:char-popular} there is a blocking structure in $G_M$. We may assume that this structure is minimal with respect to containment.

If case \emph{i)} holds, let $C=v_1-v_2-\ldots-v_k-v_1$ denote the blocking cycle ($k\geq 4$), where edge $v_1v_k$ is a blocking edge. Because of the minimality of $C$, the cycle does not contain other blocking edges.
Consider alternating walk $W=b(v_1)-v_1-v_2-\ldots-v_k-b(v_k)$ in $G_M^*$. 
If $b(v_1)=b(v_k)$ were true, nodes $v_1$ and $v_k$ would be leaf nodes of a star and they would not be connected by a blocking edge, so we have $b(v_1)\neq b(v_k)$. Thus $W$ is an alternating path in $G_M^*$ connecting two unmatched nodes.

If case \emph{ii)} holds, let $P=v_1-v_2-\ldots-v_{k-1}-v_k$ ($k\geq 4$) denote the blocking path between disjoint blocking edges $v_1v_2$ and $v_{k-1}v_k$. Because of the minimality of $P$ it does not contain other blocking edges.
Consider alternating walk $W=b(v_2)-v_2-\ldots-v_{k-1}-b(v_{k-1})$ in $G_M^*$. 
If $b(v_2)=b(v_{k-1})$ were true, nodes $v_2$ and $v_{k-1}$ would be leaf nodes of the same star and they could not be part of two disjoint blocking edges, so we have $b(v_2)\neq b(v_{k-1})$, thus $W$ is again an alternating path connecting two unmatched nodes.

Finally, if case \emph{ii)} holds, let $P=v_1-v_2-\ldots-v_{k}-w$ ($k\geq 3$) denote the blocking path between blocking edge $v_1v_2$ and unmatched node $w$. 
Because of the minimality,  $v_1v_2$ is the only blocking edge in $P$, and alternating path $W=b(v_2)-v_2-\ldots-v_{k}-u$ is in $G_M^*$ connecting unmatched nodes $b(v_2)$ and $u$. 

Now we turn to the 'if' part of the theorem.


$\Rightarrow:$

Let $w_1-v_1-\ldots-v_k-w_2$ be an alternating path connecting unmatched nodes $w_1$ and $w_2$ in $G_M^*$.
First we consider the case when there is a blocking edge between $v_1$ and $v_k$.
\begin{claim}\label{cl:cycle}
If there is a blocking edge connecting $v_1$ and $v_k$ in $G$, then there is a blocking cycle in $G_M$ through $v_1v_k$.
\end{claim}
\begin{proof}
Straightforward, since $v_1-\ldots-v_k$ is an alternating path also in $G$ connecting matching edges $v_1v_2$ and $v_{k-1}v_k$, so $v_1-v_2-\ldots-v_k-v_1$ is a blocking cycle. 
\end{proof}
Second we consider the case when one of the endpoints is $u$.
\begin{claim}\label{cl:if w_1=u}
If $w_1$ or $w_2$ is $u$, then there is a blocking structure in $G_M$.
\end{claim}
\begin{proof}
Assume $w_1=u$. Since $uv_1\in E^*$, there is an edge $xv_1\in E$ where $x$ is a node not matched by $M$.
Since $w_2\neq u$, we have $b(v_k)=w_2$ thus there is a blocking edge $v_ky$ covering $v_k$ (or a blocking edge $xy$ if $k=0$).
If $y\neq v_i$ ($1 \leq i\leq k$), path $x-v_1-\ldots-v_k-y$ is blocking,  connecting unmatched node $x$ to blocking edge $v_ky$.

If $y=v_i$, depending on the parity of $i$ we have either a blocking cycle or a blocking path in $G_M$: if $i$ is odd, we have blocking cycle $v_i-v_{i+1}-\ldots-v_k-v_i$ (note that $k$ is even), whereas if $i$ is even, we have blocking path $x-v_1-\ldots-v_{i-1}-v_i-v_k$.  
\end{proof}
Finally we investigate when cases of Claims \ref{cl:cycle}, \ref{cl:if w_1=u} do not hold.
\begin{claim}\label{cl:disjoint blockings}
If $w_1, w_2\neq u$ and $v_1v_k$ is not a blocking edge, then there exist disjoint blocking edges $y_1v_1$ and $y_2v_k$ in $G_M$.
\end{claim}
\begin{proof}
Since $w_i\neq u$, there are blocking edges $z_1v_1$ and $z_2v_k$. Note that $z_1=v_k$ and $z_2=v_1$ cannot hold because $v_1v_k$ is not a blocking edge.
If $z_1\neq z_2$, the blocking edges are disjoint, and we are done by choosing $z_i=y_i$ ($i=1,2$). 

Else $z_1=z_2$. If both $v_1$ and $v_k$ were leaves, they would be leaf nodes in the same star with middle node $z_2$, but this is not possible since $w_1\neq w_2$. So we may assume that $v_1$ is not a leaf node and it has another incident blocking edge $z_3v_1$. Then we can choose $y_1:=z_3$ and $y_2:=z_2$.
\end{proof}

\begin{claim}\label{cl:path+2edges}
If there are disjoint blocking edges $y_1v_1$ and $y_2v_k$ in $G_M$, then there is a blocking structure in $G_M$.
\end{claim}
\begin{proof}
If $y_i\neq v_j$ ($i=1,2$, $1< j < k$), then $y_1-v_1-v_2-\ldots-v_k-y_2$ is a blocking path connecting two disjoint blocking edges.

If $y_1= v_j$ and $j$ is even, then $v_1-v_2-\ldots-v_j-v_1$ is a blocking cycle through $v_1v_j$. 

If $y_1= v_j$, $j$ is odd, and $y_2\neq v_i$ ($1<i<k$), then $v_1v_j-v_{j+1}-\ldots-v_k-y_2$ is a blocking path.

Analogously, if $y_2= v_l$ and $l$ is odd, then there  is a blocking cycle $v_l-v_{l+1}-\ldots-v_k-v_l$  through $v_lv_k$, and if $l$ is even and $y_1\neq v_i $ ($1<i< k$), then there is a blocking path again.

Else we have that $y_1= v_j$ with $j$ odd and $y_2=v_l$ with $l$ even. Now if $j<l$, then path $v_1-v_j-v_{j+1}-\ldots-v_{l-1}-v_l-v_k$ is a blocking path connecting $v_1v_j$ and $v_lv_k$, whereas if $l<j$, then path $v_j-v_1-v_2-\ldots-v_{l-1}-v_l-v_k$ is blocking.
\end{proof}

Combining these claims we can prove the 'if' part of the theorem: if $w_1$ or $w_2$ is $u$, we can apply Claim \ref{cl:if w_1=u} and get a blocking structure. Also, if $v_1v_k$ is a blocking edge, there is a blocking structure by Claim \ref{cl:cycle}. If neither of the above hold, we can apply Claim \ref{cl:disjoint blockings} and then Claim \ref{cl:path+2edges}, which proves the existence of a blocking structure in this case too. This proves the 'if' part of the theorem, so the proof of the first sentence of the theorem is complete.
Since deciding whether a given matching is maximum size can be tested in linear time \cite{GT85}, we get indeed an \runtime  \;algorithm for testing the popularity of a matching, which concludes the theorem.
\end{proof}

\section{Dual witness}\label{sec: polyhedra}

In this section we investigate some consequences of the maximum matching approach for popularity testing. 
The goal is to show that an optimal solution for LP2 with restricted values can be given, summarized in the following theorem.

\begin{theorem}\label{thm:0-2}
Let a popular roommates instance be given on graph $G$. If $M$ is a popular matching in $G$, then  there exists a dual witness for LP2 with $\boldsymbol{\alpha}\in\{0,\pm 1\}^n$ and $\boldsymbol{y}\in \{0,2\}^{|\mathbb{B}|}$. Moreover, it can be assumed that odd sets of value 2 form a sub-partition of $V$.
\end{theorem}

We will prove this theorem in Subsection \ref{subsec:witness from GE} by using the Gallai-Edmonds decomposition of $G_M^*$.
An application is presented in Subsection \ref{subsec:frac} where we give a combinatorial characterization of fractional popular matchings.

\subsection{Dual witness from the Gallai-Edmonds decomposition}\label{subsec:witness from GE}

In this subsection we show that a witness described in Theorem \ref{thm:0-2} can be given
from the Gallai-Edmonds Decomposition of $G_M^*$. First we summarize some important properties of this decomposition. For more details we recommend the book of Schrijver \cite{Sch03}.

\subsubsection*{A brief introduction to the Gallai-Edmonds decomposition}

Let $G'=(V',E')$ be a graph, and consider the following partition of its nodes (we use similar notations as in \cite{Sch03} in Section 24.4b): 
\begin{itemize}
    \item let $D'$ denote the set of nodes in $G'$ that are not covered by every maximum matching,
    \item let $A'$ denote the neighbors of $D'$ in $G'$,
    \item finally let $C'$ denote the rest of the nodes in $V'$.
\end{itemize}
A graph is \textbf{factor-critical} if for every node, the deletion of that node gives a graph that has a perfect matching. 
Components in $D'$ are factor-critical, and components in $C'$ are even (\cite{Sch03} Theorem 24.7 and Corollary 24.7a).

\begin{prop}\label{prop:GE properties}
    
The following properties hold for a maximum matching $M'$:
\begin{enumerate}[i)]
    \item Every node $v\in A'$ is matched by $M'$, and $M'(v)\in D$.
    \item Every unmatched node in $G'$ is in $D'$.
    \item For every component $Z$ in $D'$ there is exactly one node not matched by $M'$ within $Z$, which we call the \textbf{root} of $Z$. There is an alternating path from the root to every node in $Z$ (because each component is factor-critical).
    \end{enumerate}
\end{prop}

\subsubsection{Dual witness}

Assume that $M$ is popular in $G$ and consider the Gallai-Edmonds decomposition of $G_M^*$: 
\begin{itemize}
    \item let $D$ denote the set of nodes in $G_M^*$ that are not covered by every maximum matching,
    \item let $A$ denote the neighbors of $D$ in $G_M^*$,
    \item finally let $C$ denote the rest of the nodes in $V^*$.
\end{itemize}

Since $M$ is popular, it is also maximum in $G_M^*$, thus all unmatched nodes (i.e. blocking nodes, star nodes and $u$) are in $D$. In addition, for every blocking node $b_v$ and node $v \in A$. Thus $b_v$ forms a $1$-element component of $D$.

Let $X$ denote the set of nodes in $V^*$ reachable on an alternating path from blocking nodes or star nodes.

\begin{corollary}\label{prop:X props}
The following properties hold for $X$:
\begin{enumerate}[a)]
    \item $u\notin X$ (otherwise $M$ is not maximum in $G_M^*$),
    \item if a matched node $v \in X$, then $M(v) \in X$,
    \item $X\cap C = \emptyset$ (from i) of Proposition \ref{prop:GE properties}),
    \item if a component $Z$ of $D$ intersects $X$, then $Z\subset X$ (from ii) of Proposition \ref{prop:GE properties}).   
    \item there is no edge between $X\cap D$ and $V^*\setminus X$ (also from ii) of Proposition \ref{prop:GE properties}),
\end{enumerate}
\end{corollary}

Now we give a dual witness $(\alpha, y)$ for the popularity of $M$ based on the structure above.  
First we define $y$. Consider the (odd) components $Z_1,\ldots,Z_k$ in  $G_M^*[X\cap D]$ that have size at least 3. 
If $Z_i\subset V$, we define $y_{Z_i}:=2$.
If $Z_i\not\subset V$, then the root of $Z_i$ is a star node $b_S$. Let $Z_i':=Z_i\setminus \{b_S\}\cup \{q\}$, where $q$ is the middle node of star $S$. Since $b_q$ is a $1$-degree unmatched node, $b_q\in D$ and  $q \in A$, so $q \notin Z_i$ thus set $Z_i'$ is also odd. We define $y_{Z_i'}:=2$.
For all other odd sets in $\mathbb{B}(G)$ we set $y$ to $0$.

Next we define $\alpha$.
For every $v\in D\cap X\cap V$ we define $\alpha_v:=-1$ and for every $v \in A \cap X$ we set $\alpha_v:=1$ (note that $A\subset V(M) \subset V$). 
Finally for every node $v \in V\setminus X$ we define $\alpha_v:=0$.

The definition of $(\alpha,y)$ is complete. 
We turn to the proof of correctness of $(\alpha,y)$, which gives a proof for Theorem \ref{thm:0-2}. 

\begin{proof}[Proof of Theorem \ref{thm:0-2}]


\begin{claim}\label{cl:subpart}
The collection of $2$-valued odd sets form a sub-partition.    
\end{claim}
\begin{proof}
The collection of $2$-valued sets is derived from the odd components of $D$, which forms a sub-partition of $V^*$. Only middle nodes of stars may be added to these sets, hence it is enough to check these middle nodes. For every star $S$ there is at most one odd component in $D$ including $b_S$, thus there is at most one $2$-valued set including the middle node of $S$, which proves the claim.    
\end{proof}

\begin{lemma}\label{lem:witness (alpha,y)}
Vector $(\alpha,y)$ is a dual witness for the popularity of $M$ in LP2. 
\end{lemma}
\begin{proof}
First we consider dual constraints corresponding to nodes (loops). 
All $\alpha_v\geq -1$, so it is enough to check nodes $v$ with $w_M(v,v)=0$. These are exactly the unmatched nodes, which all get $\alpha$ value 0.

Now we turn to dual constraints corresponding to edges.
For $-2$-weight edges the constraint is trivially met. 
\begin{claim}\label{cl:0-edges}
The dual constraint is met by $(\alpha, y)$ for $0$-weight edges.
\end{claim}
\begin{proof}
Let $vw$ be such an edge.
If both $v$ and $w$ have $\alpha$ value at least $0$, or have values $+1$ and $-1$, then the dual constraint is clearly fulfilled.

If $v$ and $w$ both have $\alpha$ value $-1$, then both are in $X\cap D$, so $v$ and $w$ are in the same odd component $Z$ of $G_M^*[X\cap D]$. Then edge $vw$ is in a $2$-valued odd set, and the dual constraint is fulfilled.

Finally we show that values $0$ and $-1$ are not possible. Assume $\alpha_v=0$ and $\alpha_w=-1$. If $v \in V^*$ were true, then edge $vw$ would connect a node in $X\cap D$ with a node in $V^*\setminus X$, a contradiction (Corollary \ref{prop:X props}).
If $v \notin V^*$ were true, then $v$ would be a node not matched by $M$, so there would be an edge $wu$ in $G_M^*$, giving a contradiction from the corollary again.
\end{proof}

\begin{claim}\label{cl:2-edges}
The dual constraint is met by $(\alpha, y)$ for $2$-weight (blocking) edges.
\end{claim}
\begin{proof}
Let $xy$ be a blocking edge. If it is not a leaf in a star, then blocking nodes $b_x$ and $b_y$ exist, and $x,y\in A$, thus $\alpha_x+\alpha_y=1+1=2$.

If $xy$ is a leaf of a star $S$ with middle node $x$ and leaf node $y$,
similarly we get that $\alpha_x=1$. If $\alpha_y=1$ also holds, then the dual constraint is fulfilled. 
Since $b_S\in D$, if $\alpha_y\neq 1$, then $\alpha_y=-1$ and $y\in D$. Thus $b_S$ and $y$ are in an odd component $Z$ of $G_M^*[X\cap D]$, and $Z'=Z\setminus \{b_S\}\cup \{x\}$ is a 2-valued odd set, and we get for the dual constraint of edge $xy$ that $y_{Z'}+\alpha_x+\alpha_y=2+1+(-1)=2$. 
\end{proof}

\begin{claim}\label{cl:total}
Dual solution $(\alpha, y)$ is optimal.
\end{claim}
\begin{proof}
We show that $(\alpha, y)$ is a $0$-weight dual solution, which proves its optimality (because matching $M$ gives a weight $0$ solution for LP). 
We sum up values of the objective function along matching edges. 
For every $v \in X\cap A$ we have $M(v)\in D$ (by Proposition \ref{prop:GE properties}), thus $\alpha_v+\alpha_{M(v)}=1+(-1)=0$.
For every $2$-valued odd set $Z$ there are $(|Z|-1)/2$ matching edges spanned by $Z$ and their endpoints have a total $\alpha$ value of $-(|Z|-1)$, which gives $y_Z(|Z|-1)/2+\sum_{vw\in M, v,w \in Z} (\alpha_v+\alpha_w)=2 \cdot (|Z|-1)/2-2 \cdot (|Z|-1)/2=0$.
All other values of $(\alpha,y)$ are zero, thus the total value is zero indeed.
\end{proof}

We have showed that $(\alpha, y)$ is a 0-weight dual solution, so it is a witness for popularity, and Lemma \ref{lem:witness (alpha,y)} is complete.
\end{proof}

Since witness  $(\alpha, y)$ defined above uses values $\boldsymbol{\alpha}\in\{0,\pm 1\}^n$ and $\boldsymbol{y}\in \{0,2\}^{|\mathbb{B}|}$, and 2-valued sets form a sub-partition by Claim \ref{cl:subpart}, this proves Theorem \ref{thm:0-2}.

\end{proof}

The following observation will be used in the next subsection to give a combinatorial characterization of fractional popularity.

\begin{observation}\label{obs:2-valued}
Let $(\alpha,y)$ be the optimal solution of LP2 derived from the Gallai-Edmonds decomposition as described above. Then there is a $2$-valued odd set if and only if there is a component $Z$ in $X\cap D$ with $|Z|\geq 3$.    
\end{observation}

\subsection{Combinatorial characterization of fractional popularity}\label{subsec:frac}

In this subsection we show an application of Theorem \ref{thm:0-2} for fractional matchings.
The notion of popularity can be generalized naturally to non-integer vectors in the following way. A \textbf{fractional matching} in $\widetilde{G}$ is a vector $p:E(\widetilde{G})\to \mathbb{R}^+$ such that $d_p^{\widetilde{G} }(v)=1$ for every $v\in V$. 
Every matching $M$ can be regarded as a $0-1$ valued fractional matching $\chi_{\widetilde{M}}$.  
A  matching $M$ is \textbf{fractional popular}\footnote{also called 'truly popular' in \cite{Kav19}} if there is no fractional matching $p$ such that $\sum_{e\in E(\widetilde{G})} w_M(e) \cdot p_e>0 $.
Kavitha \cite{Kav19} investigated fractional popular matchings and gave the following characterization.

\begin{theorem}[Kavitha \cite{Kav19}, Theorem 8.]\label{thm:Kav}
Let a popular roommates instance be given on graph $G$. A matching $M$ is fractional popular iff $M$ has a witness $(\alpha, y)$ for \emph{LP2} such that $\alpha \in \{0, \pm 1\}^{V}$ and $y=0$.
\end{theorem}

We use our results to give a combinatorial characterization of fractional popular matchings.

\begin{theorem}\label{thm:char-trula-popular}
Let a popular roommates instance be given on graph $G$ and let $M$ be a popular matching. Let sets $X,D$ be the ones derived from the Gallai-Edmonds decomposition of $G_M^*$ as described in Subsection \ref{subsec:witness from GE}.
Then the followings are equivalent:
\begin{enumerate}[i)]
    \item $M$ is also fractional popular,
    \item there is no odd component $Z$ in $G_M^*[X\cap D]$ with $|Z|\geq 3$,
    \item the following structures do not exist in $G$:
 \begin{enumerate}[a)]
    \item an odd alternating cycle $C$ through a star: $C=x-v_1-M(v_1)-\ldots-v_k-M(v_k)-x$, where $xv_1$ and $xM(v_k)$ are leafs of a star.
   \item an alternating path $P$ connecting a blocking edge $xv_1$ with the root $r$ of an odd alternating cycle $Q$: $P=x-v_1-M(v_1)-\ldots-v_k-M(v_k)=r$, and $Q=r-z_1-M(z_1)-\ldots-z_l-M(z_l)-r$, where $V(P) \cap V(Q)=\{r\}$,
 \end{enumerate}

\end{enumerate}
 
\end{theorem}
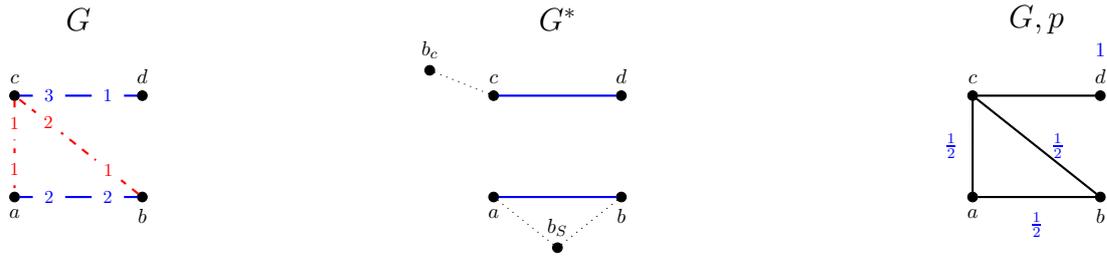
\begin{figure}[ht]
\begin{tikzpicture}[scale=0.7, transform shape]

  \pgfmathsetmacro{\d}{2.4}	
 \pgfmathsetmacro{\D}{9}
	
    \node[vertex, label=below:$a$] (A1) at (0,0) {};
	\node[vertex, label=below:$b$] (A2) at ($(A1) + (\d, 0)$) {};
	\node[vertex, label=above:$c$] (A3) at ($(A1) + (0, 0.8*\d)$) {};
	\node[vertex, label=above:$d$] (A4) at ($(A1) + (\d, 0.8*\d)$) {};

 \node[] (G) at (0.5*\d, 1.4*\d) {\LARGE{$G$}};

	\draw [thick, color=blue] (A1) -- node[edgelabel, near start] {2} node[edgelabel, near end] {2} (A2);
	\draw [thick, dashed, color=red] (A2) -- node[edgelabel, near start] {1} node[edgelabel, near end] {2} (A3);
	\draw [thick,  dashed, color=red] (A3) -- node[edgelabel, near start] {1} node[edgelabel, near end] {1} (A1);  
    \draw [thick, color=blue] (A3) -- node[edgelabel, near start] {3} node[edgelabel, near end] {1} (A4);

    \node[vertex, label=below:$a$] (A1) at (\D,0) {};
	\node[vertex, label=below:$b$] (A2) at ($(A1) + (\d, 0)$) {};
	\node[vertex, label=above:$c$] (A3) at ($(A1) + (0, 0.8*\d)$) {};
	\node[vertex, label=above:$d$] (A4) at ($(A1) + (\d, 0.8*\d)$) {};
\node[vertex, label=above:$b_c$] (bC) at ($(A1) + (-0.5*\d, \d)$) {};
\node[vertex, label=above:$b_S$] (bS) at ($(A1) + (0.5*\d, -0.4*\d)$) {};

    \node[] (G*) at ($(A1) + (0.5*\d, 1.4*\d)$) {\LARGE{$G_M^*$}};

	\draw [thick, color=blue] (A1) --  (A2);
    \draw [thick, color=blue] (A3) --  (A4);
    \draw [dotted] (A3) --  (bC);
    \draw [dotted] (A1) --  (bS);
    \draw [dotted] (A2) --  (bS);

    \node[vertex, label=below:$a$] (A1) at (\D*2,0) {};
	\node[vertex, label=below:$b$] (A2) at ($(A1) + (\d, 0)$) {};
	\node[vertex, label=above:$c$] (A3) at ($(A1) + (0, 0.8*\d)$) {};
	\node[vertex, label=above:$d$] (A4) at ($(A1) + (\d, 0.8*\d)$) {};

    \node[ label=above:\textcolor{blue}{$1$}] (p(d) at ($(A1) + (\d, \d)$) {};
	\draw [thick] (A1) --  node[ label=below:\textcolor{blue}{$\frac{1}{2}$}] {} (A2);
	\draw [thick] (A2) --  node[ label=right:\textcolor{blue}{$\frac{1}{2}$}] {} (A3);
	\draw [thick] (A3) --  node[ label=left:\textcolor{blue}{$\frac{1}{2}$}] {} (A1);  
    \draw [thick] (A3) --  (A4);

 \node[] (G) at ($(A1)+(0.5*\d, 1.4*\d)$) {\LARGE{$G, p$}};

\end{tikzpicture}
\caption{An example for the existence of structure a) in Theorem \ref{thm:char-trula-popular}. For matching $M=\{ab,cd\}$, nodes $\{a-b-b_S\}$ form an odd component in $X\cap D$ and an odd alternating cycle as well.
Thus there is an odd alternating cycle $C=c-a-b-c$ through a star in $G$ and a fractional matching $p$ more popular than $M$ exists.}\label{fig:a)}
\end{figure}

\begin{figure}[ht]
\begin{tikzpicture}[scale=0.7, transform shape]

  \pgfmathsetmacro{\d}{2.4}	
 \pgfmathsetmacro{\D}{8}
	
    \node[vertex, label=below:$a$] (A1) at (0,0) {};
	\node[vertex, label=below:$b$] (A2) at ($(A1) + (\d, 0)$) {};
	\node[vertex, label=above:$c$] (A3) at ($(A1) + (0.5*\d, 0.8*\d)$) {};

	\draw [thick, color=blue] (A1) -- node[edgelabel, near start] {2} node[edgelabel, near end] {1} (A2);
	\draw [thick, dotted] (A2) -- node[edgelabel, near start] {2} node[edgelabel, near end] {1} (A3);
	\draw [thick,  dashed, color=red] (A3) -- node[edgelabel, near start] {2} node[edgelabel, near end] {1} (A1);

    \node[] (G) at (1.15*\d, 1.3*\d) {\LARGE{$G$}};

     \node[vertex, label=below:$e$] (B1) at ($(A1) + (\d*1.3, 0)$) {};
	\node[vertex, label=below:$f$] (B2) at ($(B1) + (\d, 0)$) {};
	\node[vertex, label=above:$d$] (B3) at ($(B1) + (0.5*\d, 0.8*\d)$) {};

	\draw [thick, color=blue] (B1) -- node[edgelabel, near start] {2} node[edgelabel, near end] {2} (B2);
	\draw [thick, dotted] (B2) -- node[edgelabel, near start] {1} node[edgelabel, near end] {3} (B3);
	\draw [thick,  dotted] (B3) -- node[edgelabel, near start] {2} node[edgelabel, near end] {1} (B1);

    \draw [thick, color=blue] (B3) -- node[edgelabel, near start] {1} node[edgelabel, near end] {3} (A3);

    \node[vertex, label=below:$a$] (A1) at (\D*1.1,0) {};
	\node[vertex, label=below:$b$] (A2) at ($(A1) + (\d, 0)$) {};
	\node[vertex, label=above:$c$] (A3) at ($(A1) + (0.5*\d, 0.8*\d)$) {};
	\node[vertex, label=above:$d$] (A4) at ($(A1) + (1.8*\d, 0.8*\d)$) {};
     \node[vertex, label=below:$e$] (B1) at ($(A1) + (\d*1.3, 0)$) {};
	\node[vertex, label=below:$f$] (B2) at ($(B1) + (\d, 0)$) {};
\node[vertex, label=above:$b_c$] (bC) at ($(A1) + (-0.1*\d, \d)$) {};
\node[vertex, label=above:$b_a$] (bA) at ($(A1) + (-0.5*\d, -0.4*\d)$) {};

    \node[] (G*) at ($(A1) + (1.15*\d, 1.4*\d)$) {\LARGE{$G_M^*$}};

	\draw [thick, color=blue] (A1) --  (A2);
    \draw [thick, color=blue] (A3) --  (A4);
    \draw [thick, color=blue] (B1) --  (B2);
    \draw [dotted] (A3) --  (bC);
    \draw [dotted] (A3) --  (bC);
    \draw [dotted] (A2) --  (A3);
    \draw [dotted] (A4) --  (B1);
    \draw [dotted] (A4) --  (B2);
    \draw [dotted] (A1) --  (bA);

  \node[vertex, label=below:$a$] (A1) at (2*\D,0) {};
	\node[vertex, label=below:$b$] (A2) at ($(A1) + (\d, 0)$) {};
	\node[ label=below:\textcolor{blue}{$1$}] (A21) at ($(A1) + (\d, -0.2*\d)$) {};
	\node[vertex, label=above:$c$] (A3) at ($(A1) + (0.5*\d, 0.8*\d)$) {};

	\draw [thick] (A1) --  (A2);
	\draw [thick] (A2) --  (A3);
	\draw [thick] (A3) --  node[ label=left:\textcolor{blue}{$1$}] {}  (A1);

    \node[] (G) at ($(A1)+(1.15*\d, 1.3*\d)$) {\LARGE{$G, p$}};

     \node[vertex, label=below:$e$] (B1) at ($(A1) + (\d*1.3, 0)$) {};
	\node[vertex, label=below:$f$] (B2) at ($(B1) + (\d, 0)$) {};
	\node[vertex, label=above:$d$] (B3) at ($(B1) + (0.5*\d, 0.8*\d)$) {};

	\draw [thick] (B1) --  node[ label=below:\textcolor{blue}{$\frac{1}{2}$}] {} (B2);
	\draw [thick] (B2) --   node[ label=right:\textcolor{blue}{$\frac{1}{2}$}] {} (B3);
	\draw [thick] (B3) --   node[ label=left:\textcolor{blue}{$\frac{1}{2}$}] {} (B1);

    \draw [thick] (B3) --  (A3);

\end{tikzpicture}
\caption{An example for the existence of structure b) in Theorem \ref{thm:char-trula-popular}. For matching $M=\{ab,cd, ef\}$, alternating odd cycle $\{d-e-f-d\}$ can be reached on an alternating path from both $b_c$ and $b_a$, but the path from $b_c$ is shorter. 
Thus there is an odd alternating path $P=a-c-d$ to cycle $Q=d-e-f-d$ in $G$, and a fractional matching $p$ more poplar than $M$ exists.}\label{fig:b)}
\end{figure}
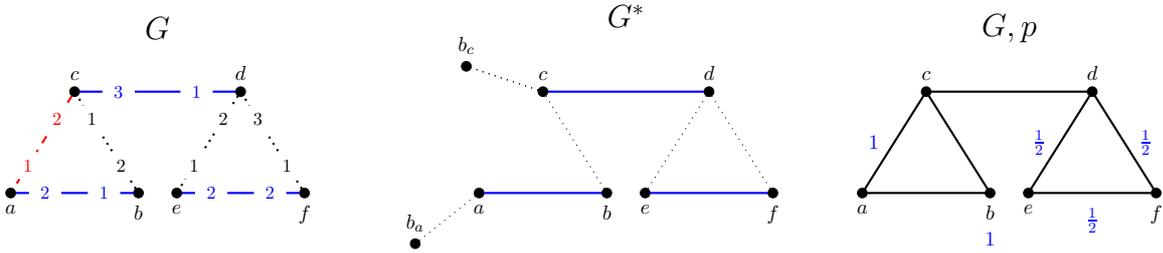

\begin{proof}[Proof of Theorem \ref{thm:char-trula-popular}]
First we prove $ii) \to i)$:  Consider the Gallai-Edmonds decomposition of $G_M^*$.    
If there is no odd component $Z$ in $G_M^*[X\cap D]$ with $|Z|\geq 3$, then by Observation \ref{obs:2-valued} there is no 2-valued odd set in witness $(\alpha,y)$ described in Subsection \ref{subsec:witness from GE}. Since $\alpha_v \in \{0,\pm 1\}$, by Theorem \ref{thm:Kav} $M$ is truly popular indeed.

Second we prove $i) \to iii)$ indirectly: assume that one of the structures a) or b) described above exists.

If an alternating cycle $C$ as described in a) exists, consider the fractional matching $p$ derived from $M$, where we replace values on cycle $C$ by constant $\frac{1}{2}$ and define $p$ to be $1$ on loop $M(x)M(x)$.
Then  $\sum_{e\in E(\widetilde{G})} w_M(e) \cdot p_e = \frac{1}{2}(2+0+\ldots+0+2)-1=1>0$ that is, $M$ is not fractional popular (see Figure \ref{fig:a)}).

If an alternating path $P$ as described in b) exists, consider the following fractional matching $p$. We switch $M$ along $P$: 
$1=p(xv_1)=p(M(v_1)v_2)=\ldots=p(M(v_{k-1})v_k)$, and $0=p(xM(x))=p(v_iM(v_i))$, and set $p$ on loop $M(x)M(x)$ to 1, while $p$ on the edges of cycle $Q$ is set to $\frac{1}{2}$. 
On all other edges $p:=\chi_{\widetilde M}$.
Then  $\sum_{e\in E(\widetilde{G})} w_M(e) \cdot p_e = 2 +0 + \ldots + 0 + (-1) + \frac{1}{2} \cdot 0=1>0$ that is, $M$ is not fractional popular  (see Figure \ref{fig:b)}).

Finally we prove $iii)\to ii)$ indirectly: let $Z$ be an odd component of size at least $3$ in $G_M^*[X\cap D]$. Since $Z$ is factor critical (Corollary 24.7a. \cite{Sch03}), there is an odd $M$-alternating cycle $C_0$ through root $r$ in $Z$: $C_0=r-v_1-M(v_1)-\ldots-v_k-M(v_k)-r$.
If the root of $Z$ is not matched by $M$, then it is a star node $b_S$ for star $S$ with middle node $x$.  Then $C:=C_0\setminus\{b_S\}\cup\{x\}:$ $C=x-v_1-M(v_1)-\ldots-v_k-M(v_k)-x$ is the alternating cycle as described in a).

If the root $r$ of $Z$ is matched by $M$, since it is in $X$, it can be reached from a star node or a blocking node $b$ on an alternating path $P_0=b-v_1-M(v_1)-\ldots-v_k-M(v_k)=r$. We may assume $P_0$ is a shortest such path.
Since $bv_1$ is an edge in $G_M^*$, there is a blocking edge $xv_1$ incident to $v_1$. Since $P_0$ is shortest, $x\notin V(P_0)$. Then $P:=x-v_1-M(v_1)-\ldots-v_k-M(v_k)=r$ and $Q:=C_0$ give a structure as in b).

\end{proof}

\section*{Open questions}

We have seen that the popularity of a matching $M$ is equivalent to the maximality of $M$ in an auxiliary graph, which motivates several questions. 
Can this approach be used to give new algorithms for finding popular matchings, or can previous methods be simplified?
Can we prove similar results for other popular structures, for example popular $b$-matchings, or popular matchings under weak preferences?

\subsection*{Acknowledgement}

The authors would like to thank Ágnes Cseh for the inspiring discussions and her valuable comments on the paper. 

The second author was in part supported by the project no. K128611, provided by the National Research, Development and Innovation Fund of Hungary. 
The first author is supported by the János Bolyai Research Scholarship of the Hungarian Academy of Science and 
by the ÚNKP-23-5-ELTE-319 New National Excellence Program of the Ministry for Culture and Innovation from the source of the National Research, Development and Innovation Fund.
This research has been implemented with the support provided by the Ministry of Innovation and Technology of Hungary from the National Research, Development and Innovation Fund, financed under the  ELTE TKP 2021-NKTA-62 funding scheme.
\bibliography{mybib}

\begin{thebibliography}{10}

\bibitem{BIM10}
P.~Bir\'o, R.~W. Irving, and D.~F. Manlove.
\newblock Popular matchings in the marriage and roommates problems.
\newblock In {\em Proceedings of CIAC '10: the 7th International Conference on
  Algorithms and Complexity}, volume 6078 of {\em Lecture Notes in Computer
  Science}, pages 97--108. Springer, 2010.

\bibitem{Chu00}
K.~S. Chung.
\newblock On the existence of stable roommate matchings.
\newblock {\em Games and Economic Behavior}, 33(2):206--230, 2000.

\bibitem{Cse17}
{\'A}.~Cseh.
\newblock Popular matchings.
\newblock {\em Trends in Computational Social Choice}, 105, 2017.

\bibitem{FKPZ19}
Y.~Faenza, T.~Kavitha, V.~Powers, and X.~Zhang.
\newblock Popular matchings and limits to tractability.
\newblock In {\em SODA'19}, pages 2790--2809, 2019.

\bibitem{GT85}
H.~N. Gabow and R.~E. Tarjan.
\newblock A linear-time algorithm for a special case of disjoint set union.
\newblock {\em Journal of Computer and System Sciences}, 30:209--221, 1985.

\bibitem{GT91}
H.~N. Gabow and R.~E. Tarjan.
\newblock Faster scaling algorithms for general graph-matching problems.
\newblock {\em Journal of the ACM}, 38(4):815--853, 1991.

\bibitem{Gar75}
P.~G{\"a}rdenfors.
\newblock Match making: assignments based on bilateral preferences.
\newblock {\em Behavioural Science}, 20:166--173, 1975.

\bibitem{GMSZ21}
S.~Gupta, P.~Misra, S.~Saurabh, and M.~Zehavi.
\newblock Popular matching in roommates setting is {N}{P}-hard.
\newblock {\em ACM Transactions on Computation Theory (TOCT)}, 13(2):1--20,
  2021.

\bibitem{HK13}
C.-C. Huang and T.~Kavitha.
\newblock Popular matchings in the stable marriage problem.
\newblock {\em Information and Computation}, 222:180--194, 2013.

\bibitem{Irv85}
R.~W. Irving.
\newblock An efficient algorithm for the ``stable roommates'' problem.
\newblock {\em Journal of Algorithms}, 6:577--595, 1985.

\bibitem{Kav14}
T.~Kavitha.
\newblock A size-popularity tradeoff in the stable marriage problem.
\newblock {\em SIAM Journal on Computing}, 43:52--71, 2014.

\bibitem{DBLP:conf/icalp/Kavitha16}
T.~Kavitha.
\newblock Popular half-integral matchings.
\newblock In I.~Chatzigiannakis, M.~Mitzenmacher, Y.~Rabani, and D.~Sangiorgi,
  editors, {\em 43rd International Colloquium on Automata, Languages, and
  Programming, {ICALP} 2016, July 11-15, 2016, Rome, Italy}, volume~55 of {\em
  LIPIcs}, pages 22:1--22:13. Schloss Dagstuhl - Leibniz-Zentrum f{\"{u}}r
  Informatik, 2016.

\bibitem{Kav19}
T.~Kavitha.
\newblock Popular roommates in simply exponential time.
\newblock In {\em 39th IARCS Annual Conference on Foundations of Software
  Technology and Theoretical Computer Science (FSTTCS)}. Schloss
  Dagstuhl-Leibniz-Zentrum fuer Informatik, 2019.

\bibitem{KMN09}
T.~Kavitha, J.~Mestre, and M.~Nasre.
\newblock Popular mixed matchings.
\newblock In S.~Albers, A.~Marchetti{-}Spaccamela, Y.~Matias, S.~E.
  Nikoletseas, and W.~Thomas, editors, {\em Proceedings of ICALP '09: the 36th
  International Colloquium on Automata, Languages and Programming}, volume 5555
  of {\em Lecture Notes in Computer Science}, pages 574--584. Springer, 2009.

\bibitem{Man13}
D.~F. Manlove.
\newblock {\em Algorithmics of Matching Under Preferences}.
\newblock World Scientific, 2013.

\bibitem{Sch03}
A.~Schrijver.
\newblock {\em Combinatorial optimization: polyhedra and efficiency},
  volume~24.
\newblock Springer Science \& Business Media, 2003.

\end{thebibliography}

\end{document}